\documentclass[bj]{imsart}
\setattribute{journal}{name}{}

\RequirePackage[OT1]{fontenc}
\RequirePackage{amsthm,amsmath,amssymb}
\RequirePackage[colorlinks,citecolor=blue,urlcolor=blue]{hyperref}
\usepackage[round]{natbib}
\usepackage{graphicx}
\usepackage{algorithmic}

\usepackage{bbm} 
\usepackage{algorithm} 
\usepackage[all]{xy}
\usepackage{color}
\usepackage{DEF-bold}
\usepackage{DEF-script}

\theoremstyle{plain}
\newtheorem{theorem}{Theorem}
\newtheorem{lemma}[theorem]{Lemma}
\newtheorem{proposition}[theorem]{Proposition}

\theoremstyle{remark}
\newtheorem{remark}[theorem]{Remark}

\newcommand{\rang}[2]{{[#1\negthinspace:\negthinspace #2]}}

\def\Pr{\mathbb{P}}
\def\Ex{\mathbb{E}}
\def\Ind{\mathbb{I}}

\def\e{{\rm e}}
\def\p{^\prime}

\def\d{{\rm d}}
\def\rh{\varrho}

\def\tmin{t^{\rm min}} 
\def\tmax{t^{\rm max}} 
\def\tobs{t^{\rm obs}}

\def\jump{J}

\def\ss{s^*}
\def\sss{s^{\dagger}}
\def\tt{t^*}

\def\gu{g^{\rm max}}
\def\gl{g^{\rm min}}
\def\Qmin{Q^{\rm min}}
\def\rmax{r^{\rm max}}
\def\qmin{q^{\rm min}}
\def\Pmin{P^{\rm min}}
\def\reg{\Phi}
\def\lsk{L^\dagger}
\def\lmax{L^{\rm max}}

\begin{document}

\begin{frontmatter}
\title{Geometric ergodicity of Rao and Teh's algorithm  for Markov jump processes {and CTBNs}}
\runtitle{Geometric ergodicity of Rao and Teh's algorithm}

\begin{aug}
 \author{\fnms{B{\L}a{\.z}ej} \snm{Miasojedow}\thanksref{a,t2}\ead[label=e1]{bmia@mimuw.edu.pl}}
\author{\fnms{Wojciech} \snm{Niemiro}\thanksref{a}\ead[label=e2]{ wniem@mimuw.edu.pl}}
\address[a]{Institute of Applied Mathematics and Mechanics, University of Warsaw\\ 
Banacha 2, 02-097 Warsaw,  Poland\\
\printead{e1},
\printead*{e2}
}

\thankstext{t2}{The work of B{\l}a{\.z}ej Miasojedow is supported by Polish National Science Center grant no.\ 2015/17/D/ST1/01198.}

  \runauthor{B. Miasojedow and W. Niemiro}

\end{aug}

\begin{abstract}
\citet{rao2012mcmc, RaoTeh2013a}  
introduced an efficient MCMC algorithm for sampling from the posterior distribution of a hidden 
Markov jump process. The algorithm is based on the idea of sampling virtual jumps. In the present paper we show that the
Markov chain generated by Rao and Teh's algorithm is geometrically ergodic. To this end we establish a geometric drift condition towards a small set.
A similar result is also  proved for a special version of the algorithm, used for probabilistic inference in Continuous Time Bayesian Networks. 
\end{abstract}


\end{frontmatter}
\section{Introduction}\label{sec: intro} 

Markov jump processes (MJP) are natural extension of Markov chains to continuous time. 
They are widely applied in modelling of the phenomena of chemical, biological,
economic and other sciences.  An important class of MJP are continuous time Bayesian networks (CTBN) 
introduced by \citet{Sch} under the name of composable Markov chains and then reinvented
by \citet{Nod1} under the current name. Roughly, a CTBN is a multivariate MJP in which the dependence structure 
between coordinates can be described by a graph. Such a graphical representation 
allows for decomposing a large intensity matrix into smaller conditional intensity matrices.   

In many applications it is necessary to consider a situation where the trajectory of a Markov jump process 
is not observed directly, only partial and noisy observations are available. Typically, 
the posterior distribution over trajectories is then analytically intractable.  In the literature there exist several approaches 
to the above mentioned problem: based on sampling 
\citep{BoysWilkKirk2008,EFK,FaSh,GolWilk2011,GolWilk2014,GolHendSher2015,Nod2,RaoTeh2013a,rao2012mcmc}, and also based on 
numerical approximations. To the best of our knowledge the most general efficient  method for a finite state space is that proposed by 
\citet{RaoTeh2013a}, and extended to a more general class of continuous time discrete systems in \citet{rao2012mcmc}. Their algorithm is based on introducing so-called 
virtual jumps and a thinning procedure for Poisson processes.

{
Recently \citet{homo} proved geometric ergodicity of Rao and Teh's algorithm in a special case of the homogeneous MJPs observed at discrete moments and 
when the virtual jumps are introduced by unifromization procedure. In the present paper we generalise results from \citet{homo} to a larger class of MJPs, 
more general observation models and also for more general class of state dependent thinning procedures. We also establish
geometric ergodicity of a Gibbs sampler for CTBN's. Geometric ergodicity is a key property of Markov chains which implies Central Limit Theorem for 
sample averages.

Note that in practice the parameters of the hidden MJP may be unknown and have to be estimated. 
Then for both Bayesian and frequentist statistical inference the Rao and Teh's algorithm can be applied as a part of more complex algorithms.
In the Bayesian approach, the Rao and Teh's algorithm can be used within a Gibbs sampler or Metropolis-Hastings algorithm which updates unknown parameters, 
according to some posterior distribution.
In the frequentist approach, the Rao and Teh's algorithm can be applied to perform E-step of Monte Carlo or stochastic approximation version of EM algorithm. 
Such extended versions of the Rao and Teh's algorithm are not considered in our paper. We assume that the probability law of
a hidden MJP is known. However we strongly believe that geometric ergodicity of the  Rao and Teh's algorithm for a given parameters of hidden process will be 
crucial in the theoretical analysis of such complex methods. }

The rest of the paper is organised as follows. In Section~\ref{sec: MJP} we briefly introduce hidden Markov jump processes, 
next in Section~\ref{sec: depthin} we recall the  dependent thinning procedure and the Rao and Teh's algorithm.
The main result is proved in Section~\ref{sec: main} and extensions for CTBN's are given in Section~\ref{sec: CTBN}.

Throughout this paper we use $p(X)$ as the generic notation for a probability density of
a random object $X$, so it may denote different functions. Set $\{n,n+1,\ldots,m\}$ is denoted by $\rang{n}{m}$ (for integer $n\leq m$). 

\section{Hidden Markov jump processes}\label{sec: MJP}

Consider a continuous time Markov process $\{X(t),{\tmin\leq t\leq \tmax}\}$ on a finite state space $\S$.
Its probability law is defined via the initial distribution $\nu(s)=\mathbb{P}(X(\tmin)=s)$ and the transition intensities 
\begin{equation}\nonumber
  Q(t;s,s^\prime)=\lim_{h\to 0}\frac{1}{h}\Pr(X(t+h)=s\p|X(t)=s)
\end{equation}
for $s,s^\prime\in\S$, $s\not=s\p$. Let $Q(t; s)=\sum_{s^\prime\neq s} Q(t; s,s^\prime)$ denote the intensity of leaving state $s$. 
In general, process $X$ can be time-inhomogeneous, that is we allow the intensities to vary in time. 
For definiteness, assume that $X$ has right continuous trajectories. We say $X$ is a Markov jump process (MJP). 

Suppose that process $X$ cannot be directly observed but we can observe some random
quantity $Y$ with probability distribution $L(Y|X)$. Let us say $Y$ is the evidence
and $L$ is the likelihood. The problem is to restore the hidden trajectory of $X$ given $Y$.  From the
Bayesian perspective, the goal is to compute/approximate/sample from the posterior 
\begin{equation}\nonumber
 p(X|Y)\propto p(X)L(Y|X).
\end{equation}
Function $L$, transition probabilities $Q$ and initial distribution $\nu$ are assumed 
to be known. 
We consider {two typical forms of noisy observation. 
In the first part of our paper we assume} that the trajectory $X([\tmin,\tmax])$ is observed independently at $k$  deterministic time points with
some random errors. Formally, we observe $Y=(Y_1,\ldots,Y_k)$ where
\begin{equation}\label{eq: obs}
 L(Y|X)=\prod_{j=1}^k L_j(Y_j|X(\tobs_j)),
\end{equation}
for some fixed known points {$\tmin\leq \tobs_1<\cdots<\tobs_k\leq \tmax$}. 
Another type of evidence is considered later in Section \ref{sec: CTBN}, in the context of CTBNs. 
{In Remarks \ref{rem: lik} and \ref{rem: obs} we mention some alternative assumptions about 
the form of evidence.} 

The obvious standing assumption in our paper is that 
$L(Y|X)>0$ happens with nonzero probability if $X$ is given by $\nu$ and $Q$. 
{It means that the hidden MJP under consideration is ``possible'', i.e.\ the data do not contradict the probabilistic model.}  

\section{Dependent thinning and Rao and Teh's algorithm}\label{sec: depthin}

The so-called ``dependent thinning'' is a useful representation of a Markov jump process in terms
of potential times of jumps and the corresponding states \citep{rao2012mcmc}. 
{The intensities are assumed to be uniformly bounded, so the process $X$ has a finite number of jumps in the bounded interval $[\tmin,\tmax]$.}
Every  trajectory 
$X([\tmin,\tmax])$ is right continuous and piecewise constant:  $X(t)=S_{i-1}$ for $T_{i-1}\leq t< T_i$, 
where random variables $T_i$ are such that $\tmin< T_1<\cdots<T_N<\tmax<T_{N+1}$. By convention, $T_0=\tmin$.  
Random sequence of states $S=(S_0,S_1,\ldots,S_N)$ such that $S_i=X(T_i)$ is called a skeleton. 
We do not assume that $S_{i-1}\not=S_i$, and therefore the two sequences 
\begin{equation}\nonumber
 \begin{pmatrix} 
  T\\S
 \end{pmatrix}
 =
 \begin{pmatrix} T_0 & T_1 & \cdots &  T_i & \cdots & T_N \\ 
S_0 & S_1 & \cdots & S_i & \cdots & S_N  \end{pmatrix}.
\end{equation}
represent the process $X$ in a redundant way: many pairs $(T,S)$ correspond to the same trajectory $X([\tmin,\tmax])$. 
Let $J=\{i \in \rang{1}{N}\colon S_{i-1}\neq S_{i}\}\cup\{0\}$, so that $T_J=(T_i\colon i \in J)$ are moments of \textit{true} jumps and
$T_{-J}=T\setminus T_J=(T_i\colon i \not\in J)$ are \textit{virtual} jumps. By a harmless abuse of notation, 
we identify increasing sequences of points in $[\tmin,\tmax]$ with finite sets. Note that the trajectory of $X$ is
uniquely defined by $(T_J,S_J)$. Let us write $X\equiv (T_J,S_J)$ and also use the notation $\jump(X)=T_J$ for the set of
true jumps.

The state-dependent thinning procedure taken from \citet{rao2012mcmc} is the following. 
We choose a function $R(t; s)\geq Q(t; s)$, interpreted as intensity of an inhomogeneous Poisson process depending on state $s\in\S$.  
The first point of this Poisson process after time $u$, say $W$, has the  probability density
\begin{equation}\label{eq: gentimes}
 f(w|u,s)=R(w;s)\exp\left[-\int\limits_{u}^{w} R(t;s) \d t\right] \qquad (w>u).
\end{equation}
Let
\begin{equation}\label{eq: genjumps}
 P(t;s,s^\prime)=\begin{cases}
          \dfrac{Q(t;s,s^\prime)}{R(t;s)}&\text{ if }s\neq s^\prime;\\
          &\\  
          1-\dfrac{Q(t;s)}{R(t;s)} & \text{ if }s=s^\prime.
         \end{cases}
\end{equation}
Sampling of $(T,S)$ then proceeds as described in Algorithm 1. 

\begin{algorithm}\label{alg: depthin}
\caption{Sampling by state-dependent thinning.}
\begin{algorithmic}
 \STATE Let ${T}_0=\tmin$ and $i=0$.
 \STATE Sample ${S}_0\sim\nu(\cdot)$.
 \WHILE{${T}_i<\tmax$}
 \STATE Let $i=i+1$.
 \STATE Sample $T_i\sim f(\cdot|T_{i-1},S_{i-1})$ \COMMENT{ given by \eqref{eq: gentimes}}.
 \STATE Sample ${S}_i\sim P({T}_i; S_{i-1},\cdot\;)$ \COMMENT{ given by \eqref{eq: genjumps}}.
 \ENDWHILE
\end{algorithmic}
\end{algorithm}

By \eqref{eq: gentimes} and \eqref{eq: genjumps}, the joint probability distribution of $(T,S)$ is the following.
\begin{equation}\label{eq: genprob}
\begin{split}
 p(T,S)&=\nu(S_0)\prod_{i=1}^N P(T_i;S_{i-1},S_i)R(T_i;S_{i-1})\exp\left[-\int\limits_{T_{i-1}}^{T_i} R(t;S_{i-1}) \d t\right] \\
       &\qquad\qquad\times  \exp\left[-\int\limits_{T_{N}}^{\tmax} R(t;S_{N}) \d t\right].
\end{split}
\end{equation}
The last part of the above expression is equal to $\Pr(T_{N+1}>\tmax|T_N,S_N)$. 
The pair $(T,S)$ produced by Algorithm 1 is a redundant representation of MJP $X$ defined by $\nu$ and $Q$ 
(probability distribution of $X$ obtains if we ``integrate out'' virtual jumps).   

\citet{rao2012mcmc} exploit dependent thinning to construct a special version of a Gibbs sampler which converges
to the posterior $p(X|Y)$. The key facts behind their algorithm are the following. First, given the trajectory $X\equiv
(T_J,S_J)$ the conditional distribution of virtual jumps $T_{-J}$ is that of the inhomogeneous
Poisson process with intensity $R(t;X(t))-Q(t;X(t))\geq 0$. Second, this distribution does not change if we introduce the likelihood.
Indeed, $L(Y|X)=L(Y|T_J,S_J)$, so $Y$ and $T_{-J}$ are conditionally independent and thus $p(T_{-J}|T_J,S_J,Y)=p(T_{-J}|T_J,S_J)$. Third,
the conditional distribution $p(S|T,Y)$ is that of a hidden discrete time Markov chain and can be efficiently sampled from using
the algorithm FFBS (Forward Filtering-Backward Sampling, \citet{CarterKohn,Fru-Sch}). 
Indeed, from \eqref{eq: genprob} and \eqref{eq: obs} it follows that
\begin{equation}\label{eq: genHMM}
 p(S|T,Y)\propto p(T,S)L(Y|T,S)\propto \nu(S_0)g_0(S_0)\prod_{i=1}^N P_i(S_{i-1},S_i) g_i(S_i),
\end{equation}
where $P_i(s,s\p)=P(T_i,s,s\p)$ is the stochastic matrix defined by \eqref{eq: genjumps} and
\begin{equation}\label{eq: gis}
\begin{split}
 g_{i-1}(s)&=R(T_i,s)  \exp\left[-\int\limits_{T_{i-1}}^{T_i} R(t;s) \d t\right]\times \prod_{j:T_{i-1}\leq \tobs_j<T_i}L_j(Y_j|s),
\qquad (i=1,\ldots,N),\\
  g_{N}(s)&=  \exp\left[-\int\limits_{T_{N}}^{\tmax} R(t;s) \d t\right]\times \prod_{j:T_{N}\leq \tobs_j<\tmax}L_j(Y_j|s). 
\end{split}
\end{equation}
Note that functions $g_i$ include not only the likelihood but also a part due to the prior distribution $p(T,S)$.    

The Rao and Teh's algorithm generates a Markov chain $X_0,X_1,\ldots,X_m,\ldots$ (where $X_m=X_m([\tmin,\tmax])$ is a trajectory of a MJP), convergent
to $p(X|Y)$. A single step, that is the rule of transition from $X_{m-1}=X$ to $X_{m}=X\p$ is described in Algorithm 2. 
\begin{algorithm}\label{alg: RT}
 \caption{Single step of Rao and Teh's algorithm.}
 \begin{algorithmic}
  \STATE {\bf input:} previous state $(T_J,S_J)\equiv X$ and observation $Y$.
  \begin{itemize}
 \item[(V)]\label{RT: V} Sample a Poisson process $V$ with intensity $R(t;X(t))-Q(t;X(t))$ on $[\tmin,\tmax]$. Let $T\p=T_J\cup V$ 
 \COMMENT{\textit{new set of potential times of jumps}}.
 \item[(S)]\label{RT: S} Draw new skeleton $S\p$ from the conditional distribution $p(\cdot|T\p,Y)$ by FFBS.  
  The new allocation of virtual and true jumps is via $J\p=\{i: S_{i-1}\p\neq S_{i}\p\}\cup\{0\}$ 
  \COMMENT{\textit{we discard new virtual jumps $T\p_{-J\p}$}}.
 \end{itemize}
  \RETURN new state  $(T\p_{J\p},S\p_{J\p})\equiv X\p$.
  \end{algorithmic}
\end{algorithm}

Convergence of the algorithm has been shown by its authors in \citet{rao2012mcmc}. It follows from the fact that the chain has the stationary distribution
$p(X|Y)$ and is irreducible and periodic, provided that $R(t;s)>Q(t;s)$.

\section{Main result}\label{sec: main}

Let $A$ be the transition kernel of the Markov chain $X_m$ generated by the Rao and Teh's algorithm. 
Let $\Pi$ be the target distribution. (It is the posterior distribution of $X$ given $Y$.  
In this paper we consider only Monte Carlo randomness, so $Y$ is fixed 
and can be omitted in notation.)  

\begin{theorem}\label{th: main} Consider a hidden MJP in which the evidence is of the form \eqref{eq: obs}.
Assume that 
\begin{enumerate}
 \item\label{as: qmin} there exists an irreducible matrix $\Qmin$ such that $\Qmin(s,s\p)\leq Q(t;s,s\p)$ for all $s,s\p\in\S$, $s\not=s\p$, $t\in[\tmin,\tmax]$,
 \item\label{as: eta} there exists $\eta>0$ such that $Q(t;s)/R(t;s)\leq 1-\eta$ for all $s\in\S$, $t\in[\tmin,\tmax]$,
 \item\label{as: rmax} there exists $\rmax<\infty$ such that $R(t;s)\leq\rmax$ for all $s\in\S$, $t\in[\tmin,\tmax]$.
 \end{enumerate}
Then the chain $X_m$ produced by the Rao and Teh's Algorithm 2 is geometrically ergodic, i.e.\ there exist
constant $\gamma<1$ and function $M$ such that for every initial trajectory $X$,
\begin{equation}\nonumber
 \Vert A^m(X,\cdot)-\Pi(\cdot)\Vert_{\rm tv} \leq \gamma^m M(X).
\end{equation}
 \end{theorem}

We begin with some auxiliary results. In Lemmas \ref{lem: forback} and \ref{lem: thin} we  
consider an inhomogeneous Markov chain $S_0,S_1,\ldots,S_n$ on a finite state space $\S$ with the joint probability distribution given by 
\begin{equation}\nonumber
  \Pr(S_{1}=s_1,\ldots,S_{n}=s_n)\propto \nu(s_0)g_0(s_0)\prod_{i=1}^n P_i(s_{i-1},s_i) g_i(s_i),
 \end{equation}
where $P_i$ are stochastic matrices and $g_i$ are non-negative functions (formula \eqref{eq: genHMM} shows that the conditional distribution of skeleton 
is of this form). Assuming that {$0< n_0\leq i<n$}, we define
 \begin{equation}\nonumber
 {P}_{i-n_0:i}(s_{i-n_0},s)=\sum_{\{s_{i-n_0+1},\dots,s_{i-1}\}}\left(\prod_{l=i-n_0+1}^{i-1}P_l(s_{l-1},s_{l})\right) P_{i}(s_{i-1},s)\;.
 \end{equation}
and
 \begin{equation}\nonumber
  {P}^g_{i-n_0:i}(s_{i-n_0},s)=\sum_{\{s_{i-n_0+1},\dots,s_{i-1}\}}\left(\prod_{l=i-n_0+1}^{i-1}P_l(s_{l-1},s_{l})g_l(s_{l})\right) P_{i}(s_{i-1},s)\;.
 \end{equation}
\begin{lemma}\label{lem: forback} Assume that 
\begin{enumerate}
 \item\label{as: xi} for some $\xi>0$ inequality $P_{i-n_0:i}(s_{i-n_0},s)\geq \xi$ holds for  all $s_{i-n_0},s\in\S$,
\item\label{as: etalem} for some $\eta>0$ inequality $P_{i+1}(s,s)\geq \eta$ holds for all $s\in\S$, 
\item\label{as: g}  {for some $\gl_l$ and $\gu_l$ we have $\gl_l\leq g_l(s)\leq \gu_l$ for all $s\in\S$, $l\in\rang{i-n_0+1}{i}$.}
\end{enumerate}
Then 
 \begin{equation}\nonumber
  \Pr(S_i=s|S_{i+1}=s)\geq\delta_i,\quad\text{where}\quad \delta_i=\frac{\xi\eta}{|\S|}{\prod_{l=i-n_0+1}^{i}\frac{\gl_l}{\gu_l}.}
 \end{equation}
\end{lemma} 
\begin{proof} We condition additionally on $S_{i-n_0}=s_{i-n_0}$ and use two-sided Markov property to obtain
\begin{align*}
 \Pr(S_i=s|S_{i-n_0}=s_{i-n_0},S_{i+1}=s)&=\frac{{P}^g_{i-n_0:i}(s_{i-n_0},s)g_i(s)P_{i+1}(s,s)}{\sum_{s\p}{P}^g_{i-n_0:i}(s_{i-n_0},s\p)g_i(s\p)P_{i+1}(s\p,s)}
\\&\geq\frac{\left(\prod_{l=i-n_0+1}^{i}{\gl_l}\right){P}_{i-n_0:i}(s_{i-n_0},s)P_{i+1}(s,s)}{\left(\prod_{l=i-n_0+1}^{i}{\gu_l}\right)\sum_{s\p}{P}_{i-n_0:i}(s_{i-n_0},s\p)P_{i+1}(s\p,s)}\\
&\geq\frac{\xi\eta}{|\S|}{\prod_{l=i-n_0+1}^{i}\frac{\gl_l}{\gu_l}\;.}
\end{align*} 
Finally let us remark that conclusion of the the lemma remains trivially true if $\gl_l=0$ for some $l\in\rang{i-n_0+1}{i}$ and thus $\delta_i=0$.  
\end{proof}

\begin{remark}
In Lemma  \ref{lem: forback} we bound from below the backward transition probability used by the FFBS algorithm. 
However, the identical inequality is true also for the forward  transition probability $\Pr(S_i=s|S_{i-1}=s)$. 
\end{remark}

\begin{remark}
In the time-homogeneous case, when $P_i=P$, the first assumption of Lemma \ref{lem: forback} is essentially equivalent to 
irreducibility and aperiodicity of matrix $P$.  Note also
that the two constants $\xi$ and $\eta$ play different roles in Rao and Teh's algorithm. 
\end{remark}

\begin{lemma}\label{lem: thin}
Let the assumptions of Lemma \ref{lem: forback} hold for all $i\in\rang{n_0}{n-1}$.  Then
\begin{equation}\nonumber
 {\Ex |J|\leq n+1-\sum_{i=n_0}^{n-1}\delta_i,} 
\end{equation}   
{where  $J=\{i\in \rang{1}{n}\colon S_{i-1}\not=S_{i}\}\cup\{0\}$}.
\end{lemma}

\begin{proof} Note that {$|J|=1+\sum_{i=0}^{n-1} \Ind( S_i\not=S_{i+1})$}.
We apply Lemma \ref{lem: forback} to each $i\in\rang{n_0}{n-1}$ to obtain $\Pr( S_i=S_{i+1}|S_{i+1})\geq \delta_i$ and
consequently $\Ex \Ind( S_i\not=S_{i+1})=\Ex \Pr( S_i\not=S_{i+1}|S_{i+1})\leq 1-\delta_i$. 
For $i<n_0$ we apply the trivial bound $\Ex \Ind( S_i\not=S_{i+1})\leq 1$. 
\end{proof}

In the next proposition we establish a geometric drift condition for the Markov chain  $X_0,X_1,\ldots,X_m,\ldots$.
Consider a single step, that is transition from $X_{m-1}=X$ to $X_{m}=X\p$. 
The dependence on the input trajectory $X$ (and also on $Y$) is implicitly assumed but indicated only when necessary.   
Recall that $|\jump(X)|$ is the number of true jumps of the trajectory $X([\tmin,\tmax])$. 

\begin{proposition}[Drift Condition]\label{pr: drift} Under the assumptions of Theorem \ref{th: main},
there exist $\delta>0$ and $c<\infty$ such that in a single step of the Rao and Teh's algorithm,
$\Ex(|\jump(X\p)||X)\leq (1-\delta)|\jump(X)|+c$.   
\end{proposition}

\begin{proof} Let us analyse what happens in both two stages {(V)} and {(S)} of Algorithm 2. The initial $X$ is fixed. In 
stage {(V)} we add a new set $V$ of potential jumps. Since
$|V|$ has the Poisson distribution with intensity $\int_{\tmin}^{\tmax} (R(t;X(t))-Q(t;X(t)))\d t$, we have 
$\Ex|V|\leq \mu:=\rmax(\tmax-\tmin)$. Thus we obtain $T\p$ with 
$\Ex(|T\p||X)\leq |\jump(X)|+\mu$. In stage {(S)} the set $T\p$ is ``thinned'' to $T\p_{J\p}$. 
Equation \eqref{eq: genHMM} shows that
conditionally, for fixed $T\p$, sampling of the new skeleton $S\p$ fulfils the conditions of Lemma \ref{lem: thin}, with $n+1=|T\p|$. Indeed,
in view of \eqref{eq: genjumps}, Assumption \ref{as: qmin} of Theorem \ref{th: main} entails Condition \ref{as: xi} of Lemma \ref{lem: forback},
{ at least for sufficiently large $n_0$ and all $i\geq n_0$. Indeed, we can choose $\xi$ and $n_0$ are such that 
$(\Pmin)^{n_0}(s,s\p)\geq \xi$, where $\Pmin$ is the stochastic matrix
with off-diagonal elements $\Pmin(s,s\p)=\Qmin(s,s\p)/\rmax$.}  
Assumption \ref{as: eta}  of Theorem \ref{th: main} entails directly Condition \ref{as: etalem} of Lemma \ref{lem: forback}. 
Moreover, the formula \eqref{eq: gis} for $g_i(s)$ includes the ``likelihood factor'' $\prod L_j(Y_j|s)$ for at most $k$ 
indices $i$, simply because there are 
$k$ points $\tobs_j$. For the remaining $n-k$ indices, by  \eqref{eq: gis}, we have two-sided bounds
$\gl\leq g_i(s)\leq \gu$, where $\gl$ and $\gu$ do not depend on $i$ and $\gl/\gu>0$. Indeed, we can choose  
$\gl=\qmin \exp[-\rmax(\tmax-\tmin)]$ and $\gu=\rmax$, where 
$\qmin=\min_s \Qmin(s)=\min_s \sum_{s\p\not= s} \Qmin(s,s\p)$. 
Assumption \ref{as: qmin} of Theorem \ref{th: main} entails $R(t;s)\geq\qmin>0$. Consequently, in the conclusion of
Lemma \ref{lem: thin}  we have 
\begin{equation}\nonumber
 \delta_i=\frac{\xi\eta}{|\S|}\left(\frac{\gl}{\gu}\right)^{n_0}=:\delta>0
\end{equation}
for at at least
$n-(k+1)n_0$ indices, with $n+1=|T\p|$ and fixed $(k+1)n_0$. For the remaining indices we put $\delta_i=0$ and thus obtain 
\begin{equation}\label{eq: thin}
 {\Ex(|J\p||T\p)\leq \left(1-{\delta}\right)|T\p|+ (k+1)n_0\delta+\delta.} 
\end{equation}
Consequently, $\Ex(|\jump(X\p)||X)=\Ex\left(\Ex(|J\p||T\p)|X\right)\leq (1-\delta)(|\jump(X)|+\mu)+(k+1)n_0\delta+\delta$. The conclusion
of the proposition follows.
\end{proof}


\begin{proposition}[Small Set Condition]\label{pr: small}
The set $\{X: |\jump(X)|\leq h\}$ is 1-small for every $h$, i.e.\ there exists a probability measure $\reg$ and a constant $\beta>0$ such that $A(X,\d X\p)\geq \beta \reg(\d X\p)$,
whenever $|\jump(X)|\leq h$.
\end{proposition}
Recall that $A$ denotes the transition kernel of the Markov chain defined via Algorithm 2. 
$\reg$ is a called a regeneration measure.

\begin{proof} 
The scheme of our proof is the following. We will define a sequence of states $\ss=(\ss_0,...,\ss_n)$ and a sequence of times
$\tt=(\tt_0,...,\tt_n)$. Both these sequences are deterministic and fixed. The regeneration measure $\reg(\d X\p)$ is described 
in terms of $\ss$ and $\tt$ as follows:
\begin{equation}\label{eq: regener}
\begin{split}
& T\p_i\sim {\rm Uniform}(\tt_{i-1},\tt_i) \text{ independently for } i=1,\ldots,n;\\
& S\p_i=\ss_i \text{ for } i=0,1,\ldots,n.
\end{split} 
\end{equation}
Trajectory $X\p$ is determined by $(T\p,S\p)$ as described in Section \ref{sec: intro}. Note that the skeleton $S\p$ is 
deterministic and random vector $(T\p_1,\ldots,T\p_n)$ has the uniform distribution on the set 
\[\mathcal{T}=\{(t_1,...,t_n)\;:\; \tt_{i-1}\leq t_{i}\leq \tt_i\text{ for } i=1,\dots,n\}\;.\]
We will show that Algorithm 2 can be equivalently executed in such a way that the resulting $X\p$ 
is distributed according to $\reg$ with probability at least $\beta>0$, provided that $|\jump(X)|\leq h$ 
($\beta$ must not depend on $X$; it will be defined in the course of our proof). 

Now we proceed to details of our construction. To define $\ss$ and $\tt$,  
let us first choose a sequence $\sss=(\sss_1,\sss_2,\dots,\sss_k)$ such that
\[\prod_{j=1}^k L_j(Y_j|\sss_j)=:\lsk>0.\]
Now we are going to use Assumptions \ref{as: qmin} and \ref{as: rmax} of Theorem 
\ref{th: main}. By irreducibility of matrix $\Qmin$ we can embed $\sss$ in a skeleton $\ss$, which has 
probability bounded below for the chain with transition matrices $P_i$ (whatever the choice of the times of jumps,
on which these matrices depend). Put differently, we define a sequence $\ss=(\ss_0,...,\ss_n)$ for some  
$n\geq k$ such that $\sss$ is a subsequence of $\ss$, $\ss_{i-1}\not=\ss_i$ and, uniformly in $t$, we have
\[\nu(\ss_0)\prod_{i=1}^{n} P(t;\ss_i,\ss_{i+1})\geq \nu(\ss_0)\prod_{i=1}^{n} \frac{\Qmin(\ss_i,\ss_{i+1})}{\rmax}=:\beta_1^*>0\;.\]
To get a sequence of times ``compatible with'' the skeleton $\ss$, we embed the sequence $\tobs=(\tobs_1,\ldots,\tobs_k)$ 
in a longer sequence $\tt=(\tt_0,\tt_1,\ldots,\tt_n)$.
More precisely, we choose a sequence  $\tmin=\tt_0<\tt_1<\cdots<\tt_n<\tmax$ such that $\ss_{i_j}=\sss_j$ implies
$\tt_{i_j}=\tobs_j$ for $j=1,\ldots,k$.

Fix $X$ with  $|\jump(X)|\leq h$. We are going to describe a special way in which Algorithm 2 can be executed. 
Note that Assumptions \ref{as: eta} and \ref{as: qmin} of Theorem \ref{th: main} ensure that
\[R(t;s)-Q(t;s)\geq \eta\qmin=:\epsilon>0.\]
In stage (V) we can independently sample two Poisson processes on the interval $[\tmin,\tmax]$, say $V^0$  and  $V^{\rm rest}$,
with intensities $\epsilon$ and $R(t;s)-Q(t;s)-\epsilon$, respectively. Next let $V=V^0\cup V^{\rm rest}$ and $T\p=\jump(X)\cup V$. 
Note that 
\[\Pr(V^0\in\mathcal{T})=:\beta^0>0\;.\]
Moreover, since $V^{\rm rest}$ is a Poisson process with intensity bounded by $\rmax$ we have 
\[\Pr(V^{\rm rest}=\emptyset)\geq \exp\left[-\rmax(\tmax-\tmin)\right]=:\beta^{\rm rest}>0\;.\] 

In stage (S) of the Algorithm 2 we construct skeleton $S\p$.
Although the actual sampling from $p(S\p|T\p,Y)$ is by FFBS, an equivalent result can be obtained via rejection sampling. 
We have $\prod_{i=0}^{|T\p|}g_i(S_i\p)\leq (\rmax)^{|T\p|}\lmax$, where $\lmax=\prod_{j=1}^k \max_s L_j(Y_j|s)$.
The rejection sampling proceeds as follows.
\begin{itemize}
 \item[(S1)] Simulate Markov chain $S\p$ (of length $|T\p|$) with transition matrices $P_i=P(T_i\p;\cdot,\cdot)$ 
  and initial distribution $\nu$, c.f.\ \eqref{eq: genHMM}.
 \item[(S2)] Accept the skeleton $S\p$ with probability $\prod_{i=0}^{|T\p|}  g_i(S\p_{i})/[(\rmax)^{|T\p|}\lmax]$.
 If the skeleton is not accepted then go to (S1).
\end{itemize}
(Of course the rejection method is highly inefficient and is considered only to clarify presentation.) 

We consider the following random events $\mathcal{E}_i$:
\begin{itemize}
 \item $\mathcal{E}_0$: in stage (V) we obtain $T\p=\jump(X)\cup V^0$ and $V^0\in \mathcal{T}$. 
 \item $\mathcal{E}_1$: in stage (S1) all points belonging to $\jump(X)$ are changed to virtual jumps, 
       while jumps at $V^0$ form the skeleton $\ss$.
 \item $\mathcal{E}_2$: in stage (S2) we accept the skeleton obtained in stage (S1).       
\end{itemize}
We can see that
\begin{itemize}
 \item $\mathcal{E}_0$ happens with probability at least $\beta^0\beta^{\rm rest}$.
 \item Given that $\mathcal{E}_0$ has happened, the probability of $\mathcal{E}_1$ is at least 
   $\beta_1^*\eta^{\jump(X)}\geq \beta_1^*\eta^{h}=:\beta_1$ (because $P(T\p_i,s,s)\geq \eta$).
 \item Given that $\mathcal{E}_0$ and $\mathcal{E}_1$ have happened, the probability of $\mathcal{E}_2$ is at least 
  \[\left(\frac{\qmin}{\rmax}\right)^{h+n}\exp\left[-\rmax(\tmax-\tmin)\right]\frac{\lsk}{\lmax}=:\beta_2,\]
  in view of \eqref{eq: gis}, because $|T\p|=|\jump(X)|+|V^0|\leq h+n$.
\end{itemize}
(Of course, all the probabilities in the above statements are conditional on $X$.) Putting everything together,
$\mathcal{E}:=\mathcal{E}_0\cap \mathcal{E}_1\cap \mathcal{E}_2$ happens with probability at least
$\beta:=\beta^0\beta^{\rm rest}\beta_1\beta_2>0$. If $\E$ happens then the output $X\p$ of Algorithm 2 is independent of
the input $X$ and has the probability distribution $\Phi(\d X\p)$ described in the beginning  of this proof. 
\end{proof}

Theorem \ref{th: main} immediately follows from Propositions \ref{pr: drift} and \ref{pr: small}, see for example \citet[Th.\ 9]{roberts2004general}. 

{
\begin{remark}\label{rem: lik}
 Note that Theorem \ref{th: main} holds also for other observation models than given by \eqref{eq: obs}. For example, we could consider the observed object $Y$ of quite general 
 nature but assume that $L(Y|X)\geq \varepsilon>0$.  The proofs of the drift condition and 
 especially of the small set condition would be then much simpler. However,  assumption \eqref{eq: obs} does not imply  $L(Y|X)\geq \varepsilon>0$,
 and we think that this latter condition is less realistic in applications.    
\end{remark}
}
\bigskip\goodbreak

\section{Continuous Time Bayesian Networks}\label{sec: CTBN}

\def\pa{{\rm pa}}
\def\ch{{\rm ch}}
\def\jul{\iota}
\def\tis{\tau}

\def\s{\underline{s}}
\def\pmin{{p^{\rm min}}}
\def\E{\mathcal{K}}

Let $(\N,\E)$ denote a directed graph with possible cycles.  We write $w\to u$ instead of $(w,u)\in\E$.
For every node $w\in\N$ consider a corresponding space $\S_w$ of possible states. 
Assume that each space $\S_w$ is finite.
We consider a continuous time homogeneous Markov process on the product
space $\S=\prod_{w\in\N} \S_w$. Thus a state $\s\in\S$ is a
configuration $\s=(s_w)=(s_w)_{w\in\N}$, where $s_w\in\S_w$. 
If $\W\subseteq\N$ then we write $s_\W=(s_w)_{w\in\W}$ for 
configuration $\s$ restricted to nodes in $\W$. We also use notation
$\S_\W=\prod_{w\in\W} \S_w$, so that we can write $s_\W\in\S_\W$. 
The set $\mathcal{W}\setminus\{w\}$ will be denoted by $\mathcal{W}-w$ and 
$\N\setminus\{w\}$ simply by $-w$.  We define the set of parents of node $w$ by 
\[\pa(w)=\{u\in\N\;:\;u\to w\}\;,\]
and we define the set of children of node $w$ by   
\[\ch(w)=\{u\in\N\;:\;w\to u\}\;.\]
Suppose we have a family of functions
$Q_w:\S_{\pa(w)}\times(\S_w\times \S_w)\to[0,\infty)$.
For fixed $c\in \S_{\pa(w)}$, we consider $Q_w(c;\cdot,\cdot\;)$ as a  conditional intensity matrix 
(CIM) at node $w$ (only off-diagonal elements of this matrix
have to be specified, the diagonal ones are irrelevant).
The state of a CTBN at time $t$ is a random element $X(t)$ of the space 
$\S$ of configurations. Let $X_w(t)$ denote its $w$th coordinate. 
The process $\left\{X_w(t)_{w\in\N},t\geq 0\right\}$ 
is assumed to be Markov and its evolution can be described 
informally as follows. Transitions at node $w$ depend on the current 
configuration of the parent nodes. If the state
of some parent changes, then node $w$ switches to other transition 
probabilities.  Formally, CTBN is a \textit{time-homogeneous} MJP with transition intensities given by  
\begin{equation*}
    Q(\s,\s\p)=
          \begin{cases}
             Q_w(s_{\pa(w)},s_w,s_w\p) & \text{if $s_{-w}=s_{-w}\p$ and $s_{w}\not=s_{w}\p$ for some $w$;} \\        
              0       &  \text{if $s_{-w}\not=s_{-w}\p$ for all $w$,}
          \end{cases}
\end{equation*}
for $\s\not=\s\p$. Define also $Q_w(c;s)=\sum_{s\p\not=s}Q_w(c;s,s\p)$ for $c\in\S_{\pa(w)}$, $s\in\S_{w}$. 

For a CTBN, the density  of sample path  $X=X([\tmin,\tmax])$ in a bounded time interval $[\tmin,\tmax]$ decomposes as follows:
\begin{equation}
 \label{eq:densCTBN}
 p(X)=\nu(X(\tmin))\prod_{w\in\mathcal{N}}p(X_w\Vert X_{\pa(w)})\;,
\end{equation}
where $\nu$ is the initial distribution on $\X$ and $p(X_w\Vert X_{\pa(w)})$ is the density of piecewise 
homogeneous MJP with intensity matrix equal to $Q_w(c;\cdot,\cdot\;)$ in every time sub-interval 
such that $X_{\pa(w)}=c$.  Formulas for the density of CTBN appear e.g.\ in \citet[Sec.\ 3.1]{Nod2},
\citet[Eq.\ 2]{FaXuSh}, \citet[Eq.\ 1]{FaSh} and \citet{CTBNMet2014}.
Our notation $p(X_w\Vert X_{\pa(w)})$ is consistent with the notion of ``conditioning by intervention'', 
see \e.g.\ \citet{Lauritzen01causalinference}. Indeed, $p(X_w\Vert X_{\pa(w)})$ is the density of the process $X_w$ at note $w$ 
under the assumption that the sample paths at the parent nodes $X_{\pa(w)}$ are fixed and $X_w(\tmin)$ is given, see e.g.\ 
\cite{CTBNMet2014}, for details. 
Below we explicitly write an expression for $p(X_w\Vert X_{\pa(w)})$. We need the following notations:
\begin{itemize}
  \item[] Let $\jul_w^{X}(c;\;s,s\p)$ denote the number of jumps from 
$s\in\S_w$ to $s\p\in\S_w$  at node $w$, 
which occurred when the parent configuration was $c\in\S_{\pa(w)}$.
\item[] Let $\tis_w^{X}(c;\; s)$ be the length of time that node $w$ spent in state $s \in \S_w$  when the parent configuration was $c\in\S_{\pa(w)}$.   
\end{itemize}
With these notations we can write 
\begin{equation}\label{eq: cbi}
\begin{split}
       p({X}_w\Vert {X}_{\pa(w)})&=
             \bigg\{\prod_{c\in\S_{\pa(w)}}
                      \prod_{s\in\S_w} \prod_{s\p\in\S_w\atop s\p\not=s} 
                      Q_w(c;\; s,s\p)^{\jul_w^{{X}}(c;\; s,s\p)}\bigg\}\\
             &\times\bigg\{\prod_{c\in\S_{\pa(w)}}
                       \prod_{s\in\S_w} \exp\left[-Q_w(c;\; s)
                              \tis_w^{X}(c;\; s)\right]\bigg\}.
\end{split}
\end{equation}
Let us also write $p^{\rm jump}({X}_w\Vert {X}_{\pa(w)})$ and $ p^{\rm stay}({X}_w\Vert {X}_{\pa(w)})$ for 
the first and second expression in \eqref{eq: cbi}, respectively, to facilitate future references. 

The problem of probabilistic reasoning for a CTBN can be formulated as follows. Assume that the available  evidence is the \textit{complete} 
observation of some nodes, say $X_\O([\tmin,\tmax])$ for some set $\O\subset\N$. We are to compute the posterior distribution over unobserved nodes, i.e.\ on 
the trajectories of $X_{\W}([\tmin,\tmax])$, where $\W=\N\setminus\O$. The basic idea, proposed in \citep{RaoTeh2013a} is to use ``Algorithm 2 within Gibbs sampler".
Let us fix a node $w\in\W$. By \eqref{eq:densCTBN}, the full conditional distribution is the following.
\begin{equation}
 \label{eq: fullconditionals}
p(X_w|X_{-w})\propto \nu(X_w(\tmin)|X_{-w}(\tmin)) p(X_w\Vert X_{\pa(w)})\prod_{u\in \ch(w)}p(X_u\Vert X_{\pa(u)})\;.
\end{equation}
The density $\nu(X_w(\tmin)|X_{-w}(\tmin))p(X_w\Vert X_{\pa(w)})$ corresponds to a piecewise homogeneous Markov process and can be treated as the prior distribution in
Algorithm 2. The expression $\prod_{u\in \ch(w)}p(X_u\Vert X_{\pa(u)})$
can be treated as likelihood. 
In Algorithm 2 we can use  ,,instrumental'' intensities $R_w(t;s)$ different for different nodes $w$, and possibly 
time-inhomogeneous (in practical implementations however, $R_w$ will usually be time-homogeneous).
\begin{algorithm}[H]\label{alg: CTBN}
 \caption{Random scan Gibbs sampler for CTBN.}
 \begin{algorithmic}
  \STATE {\bf input:} previous trajectory $X_{\W}$ and evidence $X_{\O}$.
  \begin{itemize}
 \item[] Choose at random {$w\in\W$ (according to some probability distribution on $\W$).}  
 \item[] Update $X_w$ to $X_w\p$. Apply Algorithm 2 with the target distribution $p(X_w|X_{-w})$ given by 
  \eqref{eq: fullconditionals} 
  \COMMENT{ \textit{$X_{-w}$ includes all nodes $u\in\O$ as well as $u\in\W$, $u\not=w$ }}.
 \end{itemize}
  \RETURN new trajectory  $X\p_{\W}=(X_{\W-w},X_w\p)$.
  \end{algorithmic}
\end{algorithm}

\begin{theorem}\label{th: ctbn} Consider a CTBN in which we observe trajectories $X_\O([\tmin,\tmax])$.
Assume that 
\begin{enumerate}
 \item\label{as: qminctbn} for every $w\in\W=\N\setminus\O$ there exists an irreducible matrix $\Qmin_w$ such that 
   $\Qmin_w(s,s\p)\leq Q_w(c;s,s\p)$ for all $s,s\p\in\S_w$, $s\not=s\p$, $c\in\S_{\pa(w)}$,
 \item\label{as: etactbn} there exists $\eta>0$ such that $Q_w(c;s)/R_w(t;s)\leq 1-\eta$ for all $w\in\W$, $s\in\S_w$,
   $c\in \S_{\pa(w)}$ and $t\in[\tmin,\tmax]$,
 \item\label{as: rmaxctbn} there exists $\rmax<\infty$ such that $R_w(t;s)\leq\rmax$ for all $w\in\W$, $s\in\S_w$,
 $t\in[\tmin,\tmax]$,
 \item\label{as: supp} for every $w\in\N$, the supports of $Q_w(c;\cdot,\cdot)$ do not depend on the parent configuration $c\in\S_{\pa(w)}$, i.e.
  $Q_w(c;s,s\p)>0$ implies  $Q_w(c\p;s,s\p)>0$.
 \end{enumerate}
Then the chain $(X_{\W})_m$ produced by Algorithm 3 is geometrically ergodic. 
 \end{theorem}

As in Section \ref{sec: main}, we will show a drift condition towards a small set. The Lyapunov function in the CTBN setting 
will be the global number of jumps $|\jump(X_{\W})|=\sum_{w\in\W} |\jump(X_w)|$.  
Let us begin with an elementary fact, needed in the proof of the drift condition.

\begin{lemma}\label{lem: lagrange} Let $0\leq\rh<1$. For any $x_1,\ldots,x_n$ such that $\sum_{i=1}^n x_i\leq nb$,
\[\sum_{i=1}^n \rh^{x_i}\geq n\rh^b.\]
\end{lemma}
\begin{proof} Without loss of generality we can assume that $\sum_{i=1}^n x_i= nb$.
 Consider the following  constrained optimisation problem:
 \[\textrm{minimize}\: \sum\rh^{x_i}\]
subject to $\sum x_i= nb$. The corresponding Lagrange function is
 \[\L(x,\lambda)=\sum\rh^{x_i}+\lambda\left(\sum x_i-nb\right)\]
 and its partial derivatives are
 \[\frac{\partial }{\partial x_i}\L(x,\lambda)=\rh^{x_i}\log\rh+\lambda\;.\]
 Therefore the minimizer is $x_1=x_2=\cdots=x_n=b$ and the conclusion follows.
\end{proof}

\begin{proposition}[Drift Condition]\label{pr: DriftCTBN}
 Under the assumptions of Theorem \ref{th: ctbn},
there exist $\eps>0$ and $c<\infty$ such that in a single step of Algorithm \ref{alg: CTBN},
$\Ex(|\jump(X_\W\p)||X)\leq (1-\eps)|\jump(X_\W)|+c$.   
\end{proposition}

\begin{proof}
For a given $X_\W$, there always exists a node $w\in\W$ such that $|\jump(X_w)|\geq |\jump(X_\W)|/|\W|$, e.g.\ a node with maximum number of jumps.
From now on, node $w$ is fixed. We will prove that for some $c_0$ and $\eps_0>0$,
\begin{equation}\label{eq: localdrift}
 \Ex(|\jump(X_w\p)||X,\text{node $w$ is updated})\leq (1-\eps_0)|\jump(X_w)|+c_0.
\end{equation}
This is a ``local version'' of the drift condition. The conclusion of the proposition will easily follow from
\eqref{eq: localdrift}.  Indeed, for every node $u$ we have $\Ex(|\jump(X_u\p)||X,\text{node $u$ is updated})$ $\leq |\jump(X_u)|+\mu$, 
with $\mu:=\rmax(\tmax-\tmin)$,  as in the proof of Proposition \ref{pr: drift}. 
Let $\pmin$ denote the minimum probability of selecting a node for update in Algorithm 3. 
The singled out node $w$ is chosen with probability at least $\pmin$. Therefore 
 \begin{equation}\nonumber
 \begin{split}
  \Ex(|\jump(X_\W\p)||X)&\leq \pmin \left((1-\eps_0) |\jump(X_w)|+c_0+|\jump(X_{\W-w})|\right)+(1-\pmin)\left(|\jump(X_\W)|+\mu\right)\\
                     &\leq \pmin \left((1-\eps_0)|\jump(X_w)|+|\jump(X_\W)|-|\jump(X_w)|\right) +(1-\pmin)|\jump(X_\W)|+c_0+\mu\\
                    &\leq |\jump(X_\W)|-\pmin\eps_0|\jump(X_w)|+c_0+\mu\\
                     &\leq |\jump(X_\W)|\left(1-\pmin\frac{\eps_0}{|\W|}\right)+c_0+\mu,
   \end{split}
 \end{equation}
which is the desired conclusion.

It remains to show \eqref{eq: localdrift}. Since node $w$ is fixed, we will omit subscript $w$ in notation whenever the context permits.   
The reasoning leading to \eqref{eq: localdrift}  is similar to the proof of  Proposition \ref{pr: drift},  but more delicate  
due to a different form of the likelihood. The role of observed $Y$ is now played by $X_{-w}$. We can write the conditional distribution of
the skeleton $S$ (at node $w$) given times of possible jumps $T$ (at node $w$) in the same form as \eqref{eq: genHMM}, namely 
\begin{equation}\label{eq: CTBNHMM}
 p(S|T,X_{-w})\propto \tilde\nu(S_0)g_0(S_0)\prod_{i=1}^N P_i(S_{i-1},S_i) g_i(S_i),
\end{equation}
where $N+1=|T|$, {$\tilde\nu(X_w(\tmin)):=\nu(X_w(\tmin)|X_{-w}(\tmin))$} and $P_i$s are defined by $P_i(s,s\p)=P(T_i,s,s\p)$, as in \eqref{eq: genHMM},
with
\begin{equation}
 P(t;s,s^\prime)=\begin{cases}
          \dfrac{Q_w(X_{\pa(w)}(t);s,s^\prime)}{R_w(t;s)}&\text{ if }s\neq s^\prime;\\
          &\\  
          1-\dfrac{Q_w(X_{\pa(w)}(t);s)}{R_w(t;s)} & \text{ if }s=s^\prime.
         \end{cases}
\end{equation}
(The ``prior distribution'' at node $w$ is that of a piecewise-homogeneous MJP). 
However, the expressions for the $g_i$s are different than \eqref{eq: gis}. We now have
\begin{equation}\label{eq: gisCTBN}
\begin{split}
 g_{i-1}(s)&=R_w(T_i,s)  \exp\left[-\int\limits_{T_{i-1}}^{T_i} R_w(t;s) \d t\right]\times  \prod_{u\in\ch(w)}p(X_u^{(i)}\Vert X_{\pa(u)}^{(i,w,s)}),
\: (i=1,\ldots,N),\\
  g_{N}(s)&=  \exp\left[-\int\limits_{T_{N}}^{\tmax} R_w(t;s) \d t\right]\times \prod_{u\in\ch(w)}p(X_u^{(N)}\Vert X_{\pa(u)}^{(N,w,s)}), 
\end{split}
\end{equation} 
where $X^{(i)}=X([T_{i-1},T_i))$ is the process restricted to an inter-jump-at-$w$ interval and $X_{\pa(u)}^{(i,w,s)}$ denotes the trajectories
of $X_{\pa(u)}^{(i)}$ with $X_w^{(i)}$ replaced by $s\in\S_w$. The likelihood parts  $p(\cdot\Vert\cdot)$ in equation \eqref{eq: gisCTBN}
can be decomposed into $p(\cdot\Vert\cdot)=p^{\rm jump}(\cdot\Vert\cdot) \times p^{\rm stay}(\cdot\Vert\cdot)$, c.f.\ \eqref{eq: cbi}. 
If we write 
\begin{equation}\nonumber
 g_i(s)=h_i(s) \prod_{u\in\ch(w)}p^{\rm jump}(X_u^{(i)}\Vert X_{\pa(u)}^{(i,w,s)}),
\end{equation}
then $h_i$ is easy to bound (at least qualitatively), because by \eqref{eq: cbi} we have
\begin{equation}\nonumber
  \exp\left[-\rmax(\tmax-\tmin)|\N|\right]\leq \prod_{u\in\ch(w)}p^{\rm stay}(X_u^{(i)}\Vert X_{\pa(u)}^{(i,w,s)})\leq 1.
\end{equation}
The part of $h_i$ which corresponds to the prior can be bounded analogously as in the proof of Proposition \ref{pr: drift} and thus 
we obtain 
\begin{equation}\label{eq: hiBound}
 \qmin \exp\left[-\rmax(\tmax-\tmin)|\right]\leq h_i(s)\leq \rmax,
\end{equation}
with $\qmin=\min_w\min_s \Qmin_w(s)$, c.f.\ Assumption \ref{as: qminctbn} of Theorem \ref{th: ctbn}.

We are now left with a task of bounding the  expression with $p^{\rm jump}$. This is more difficult, because this part of the likelihood depends on
$|\jump(X_\W)|$. By Assumptions \ref{as: etactbn} and
\ref{as: rmaxctbn} of Theorem \ref{th: ctbn}, we have $Q_u(c;s,s\p)\leq \rmax$ for all $s\not=s\p$, $s,s\p\in\S_u$, $c\in\S_{\pa(u)}$, for all $u$.  
To obtain a lower bound, we define $\qmin_+$ as the minimum of \textit{nonzero} values of $Q_u(c;s,s\p)$, $s\not=s\p$,
$s,s\p\in\S_u$, $c\in\S_{\pa(u)}$, all $u$. 
From \eqref{eq: cbi} it follows that
\begin{equation}\label{eq: jumpBound}
  \left(\qmin_+\right)^{|\jump(X_{\ch(w)}^{(i)})|}\leq \prod_{u\in\ch(w)}p^{\rm jump}(X_u^{(i)}\Vert X_{\pa(u)}^{(i,w,s)})\leq 
\left(\rmax\right)^{|\jump(X_{\ch(w)}^{(i)})|}.
\end{equation}
Note that Assumption \ref{as: supp} of Theorem \ref{th: ctbn} is needed to justify the lower bound above. Indeed, under the obvious assumption that $X$ is possible, i.e.\ 
$p(X)>0$, the jumps of $X_u^{(i)}$ are possible under the configuration $X_{\pa(u)}^{(i)}$.
By Assumption \ref{as: supp} of Theorem \ref{th: ctbn}, 
they must be possible also under the configuration $X_{\pa(u)}^{(i,w,s)}$. Combining \eqref{eq: hiBound} and  \eqref{eq: jumpBound} we obtain
 \begin{equation}\nonumber
  a_1\rh_1^{|\jump(X_{\ch(w)}^{(i)})|}\leq g_i(s)\leq a_2\rh_2^{|\jump(X_{\ch(w)}^{(i)})|}
 \end{equation}
for some constants $a_1,a_2$ and $\rh_1,\rh_2$ which depend only on the parameters of the network and on the instrumental intensity $R$ (they depend neither on $w$ nor on $i$).  
Of course, $|\jump(X_{\ch(w)}^{(i)})|=|\jump(X_{\ch(w)\cap \W}^{(i)})|+|\jump(X_{\ch(w)\cap\O}^{(i)})|$. Since $X_\O$ is fixed, 
the parts with the exponent ${|\jump(X_{\ch(w)\cap\O}^{(i)})|}$ can be absorbed in constants, which leads to the bound
 \begin{equation}\label{eq: giCTBN}
  \tilde a_1 \rh_1^{|\jump(X_{\ch(w)\cap \W}^{(i)})|}\leq g_i(s)\leq \tilde a_2\rh_2^{|\jump(X_{\ch(w)\cap \W}^{(i)})|}.
 \end{equation}

Assume that stage (V) of Algorithm 2 has been completed, resulting in a new set $T\p$ of potential times of jumps at node $w$. 
Just as in the proof of Proposition \ref{pr: drift}, we infer that $\Ex(|T\p||X,\text{$w$ is updated})\leq |\jump(X_w)|+\mu$, where $\mu=\rmax(\tmax-\tmin)$. 
Now consider stage (S) of Algorithm 2. After sampling a new skeleton $S\p$ at node $w$, the set $T\p$  is ``thinned'' of to the set of true jumps
$\jump(X_w\p)$. We are to bound $|\jump(X_w\p)|$ from above.   
We are going to apply Lemmas \ref{lem: forback} and \ref{lem: thin} to the skeleton chain
$S\p$ at node $w$ which has the probability distribution $p(S\p|T\p,X_{-w})$ given by
\eqref{eq: CTBNHMM}. {We now decompose the process $X$ into inter-jump parts according to $T\p$ and use the notation $X^{(i)}:=X([T_{i-1}\p,T_i\p))$.}     
From \eqref{eq: giCTBN} we infer that the conclusion of Lemma \ref{lem: forback} holds with
 \begin{equation}\label{eq: rhojump}
  \delta_i=a\rh^{|\jump(X_{\ch(w)\cap\W}^{(i)})|}
 \end{equation}
for some constants $a$, $\rh$, $n_0$ and for $i\geq n_0$. {Indeed, we can choose 
$n_0$ and $\xi$ such that $(\Pmin_w)^{n_0}(s,s\p)\geq \xi$, where $\Pmin_w$ is the stochastic matrix
with off-diagonal elements $\Pmin_w(s,s\p)=\Qmin_w(s,s\p)/\rmax$, c.f. Assumption \ref{as: qminctbn} of  Theorem \ref{th: ctbn}.
Then put $\rh=(\rh_2/\rh_1)^{n_0}$ and $a=(\tilde a_2/\tilde a_1)\xi\eta/|\S_w|$, where $\eta$ appears in Assumption
\ref{as: etactbn} of Theorem \ref{th: ctbn}.  }

We are now prepared to use Lemma \ref{lem: lagrange}. To verify its assumption, it is necessary to bound the sum of the exponents in \eqref{eq: rhojump}. 
Recall that  $\jump(X_u^{(i)})$ is the set of jumps at $u$ in the time interval between consecutive potential jumps at $w$, i.e.\ $[T\p_{i-1},T\p_i)$.  
Therefore
\begin{equation}\nonumber
 \sum_{i=n_0}^{|T\p|-2} |\jump(X_{\ch(w)\cap\W}^{(i)})|\leq |\jump(X_{\W})|\leq |\jump(X_{w})|\cdot |\W|\leq |T\p|\cdot |\W|\leq  2|T\p-n_0-1||\W|,
\end{equation}
where the last inequality holds for  $|\jump(X_{w})|\geq 2(n_0+1)$, 
Lemma \ref{lem: lagrange} implies that
\begin{equation}\nonumber
 \sum_{i=n_0}^{|T\p|-2}  \delta_i\geq a(|T\p|-n_0-1)\rh^{2|\W|}.
\end{equation}
Lemma \ref{lem: thin} implies that 
\begin{equation}\nonumber
 \Ex (|\jump(X_w\p)||X,\text{$w$ is updated})\leq  |T\p|-\sum_{i=n_0}^{|T\p|-2}\delta_i\leq |T\p|(1-a\rh^{2|\W|})+a(n_0+1).
\end{equation}
It is now enough to use $\Ex(|T\p||X,\text{$w$ is updated})\leq |\jump(X_w)|+\mu$ to complete the proof.
\end{proof}

{Note that  the evidence, i.e.\  the trajectory $X_\O([\tmin,\tmax])$ is considered as fixed. 
In particular the constants in Proposition \ref{pr: DriftCTBN} may depend on $X_\O$.  }


The following proposition is an analogue of Proposition \ref{pr: small} in the CTBN setting. Now $A$ denotes the transition kernel of the Markov chain 
defined via Algorithm 3.

\begin{proposition}\label{pr: smallCTBN}
The set $\{X_\W: |\jump(X_\W|)|\leq h\}$ is $|\W|$-small for every $h$, i.e.\ there exists a probability measure $\reg$ and a constant $\beta>0$ such that 
$A^{|\W|}(X_\W,\d X\p_\W)\geq \beta \reg(\d X\p_\W)$,
whenever $|\jump(X_\W)|\leq h$.
\end{proposition}

\begin{proof}
In contrast to the drift condition, the proof of the small set condition is easier in the present setting.
Under the assumptions of Theorem \ref{th: ctbn}, for the regeneration measure $\Phi$ we can take the measure concentrated at 
a deterministic, \textit{constant} trajectory $\ss_\W=(\ss_u,u\in\W)\in \S_W$. The only requirement is that $\tilde\nu(\ss_\W)>0$, 
where $\tilde\nu$ is the posterior initial distribution of $X_\W(\tmin)$, given $X_\O(\tmin)$. We are to bound from below the probability
that $X_\W\p(t)=\ss_\W$, for $t\in[\tmin,\tmax]$, where $X_\W\p$ is the result of $|\W|$ steps of Algorithm 3, starting from an arbitrary $X_\W$ such that
$|\jump(X_\W)\leq h$. 

With probability at least $(\pmin)^{|\W|}$, Algorithm 3  in $|\W|$ steps will visit and update
all nodes belonging to $\W$. Let us now consider a single step, in which $w\in\W$ is updated via Algorithm 2. 
In the rest of the proof $w$ is arbitrary but fixed. $(T\p,S\p)$ denote the times of potential jumps and the skeleton of $X_w\p$. 
Since Assumptions \ref{as: qminctbn}, \ref{as: etactbn}  and \ref{as: rmaxctbn}
of Theorem \ref{th: ctbn} are, for a fixed $w$, essentially the same as Assumptions \ref{as: qmin}, \ref{as: eta}  and \ref{as: rmax}
of Theorem \ref{th: main}, the reasoning is similar as in the proof of Proposition \ref{pr: small}.
We assume that stage (S) of Algorithm 2 is executed in a way described in that proof, via rejection sampling.
We consider the following random events $\mathcal{E}_i$:
\begin{itemize}
 \item $\mathcal{E}_0$: in stage (V) we obtain $V=\emptyset$ so $T\p=\jump(X_w)$. 
 \item $\mathcal{E}_1$: in stage (S1) all points belonging to $\jump(X_w)$ are changed to virtual jumps.
 \item $\mathcal{E}_2$: in stage (S2) we accept the skeleton obtained in stage (S1).       
\end{itemize}
It is easy to obtain the following lower bounds.
\begin{itemize}
 \item $\mathcal{E}_0$ happens with probability at least $\exp\left[-\rmax(\tmax-\tmin)\right]=:\beta_0$, 
 because $V$ is a Poisson process with intensity $R(t;s)-Q(t;s)\leq \rmax$.
 \item Given that $\mathcal{E}_0$ has happened, the probability of $\mathcal{E}_1$ is at least 
   $\eta^{|\jump(X_w)|}\geq \eta^{h}=:\beta_1$ (because $P(T\p_i,s,s)\geq \eta$).
 \item Given that $\mathcal{E}_0$ and $\mathcal{E}_1$ have happened, the probability of $\mathcal{E}_2$ is at least 
 \[\left(\frac{\tilde a_1 \rh_1}{\tilde a_2 \rh_2}\right)^h=:\beta_2,\]
  where $\tilde a_1,\rh_1,\tilde a_2, \rh_2$ are constants appearing in \eqref{eq: giCTBN}. Indeed, the acceptance
  criterion can be $\prod g_i(S_i\p)\leq (\tilde a_2 \rh_2)^h$ and it is always true that  $\prod g_i(S_i\p)\geq (\tilde a_2 \rh_1)^h$ 
  (by \eqref{eq: giCTBN};  note that  
  $\sum_i|\jump(X_{\ch(w)\cap\W}^{(i)})|\leq|\jump(X_\W)|\leq h$ and also the number of summands is at most $|\jump(X_w)|\leq h$).
  
\end{itemize}
Of course, if $\mathcal{E}_0\cap \mathcal{E}_1\cap \mathcal{E}_2$ happens then $X_w\p(t)=\ss_w$ for $t\in[\tmin,\tmax]$.  
Putting everything together, we get $X\p_w([\tmin,\tmax])=\ss_w$ with probability at least 
$\beta:=\pmin\beta_0\beta_1\beta_2>0$. 
In $|\W|$ steps we get $X\p_\W([\tmin,\tmax])=\ss_\W$  with probability at least
$\tilde\nu(\ss_\W)\beta^{|\W|}$. The proof is complete. 
\end{proof}

Theorem \ref{th: ctbn} follows from Propositions \ref{pr: DriftCTBN} and \ref{pr: smallCTBN}. 

{

\begin{remark}\label{rem: const} In this paper the focus is on qualitative results. The constants in our  bounds are chosen in a way
which makes presentation clearer, and we did not attempt to optimize them.
\end{remark}

\begin{remark}\label{rem: obs}
For clarity of presentation we have proved  Theorem \ref{th: ctbn} under the  assumption that some nodes of CTBN are fully observed. However, by minor 
modification of the proofs we can establish geometric ergodicity of the Rao and Teh's algorithm in a more general case. Our results remain true if we assume
that some nodes are only partially observed with random noise at discrete moments, just as in \eqref{eq: obs}.
Clearly, for a drift condition,  the likelihood part can be treated in the same way as in proof of Proposition~\ref{pr: drift}. For a small set condition, 
we can repeat the construction from the proof of 
Proposition~\ref{pr: small} for every node $w\in\N$, and then define  the regeneration measure for whole network as a product of regeneration measures for
single nodes. The proofs of the key propositions in the more general case are not essentially different but become notationally complicated and awkward.
For this reason they are omitted.
\end{remark}
}

\bibliographystyle{imsart-nameyear}

\bibliography{refs}

\begin{thebibliography}{18}

\bibitem[\protect\citeauthoryear{Boys, Wilkinson and
  Kirkwood}{2008}]{BoysWilkKirk2008}
\begin{barticle}[author]
\bauthor{\bsnm{Boys},~\bfnm{Richard~J}\binits{R.~J.}},
  \bauthor{\bsnm{Wilkinson},~\bfnm{Darren~J}\binits{D.~J.}} \AND
  \bauthor{\bsnm{Kirkwood},~\bfnm{Thomas~BL}\binits{T.~B.}}
(\byear{2008}).
\btitle{Bayesian inference for a discretely observed stochastic kinetic model}.
\bjournal{Statistics and Computing}
\bvolume{18}
\bpages{125--135}.
\end{barticle}
\endbibitem

\bibitem[\protect\citeauthoryear{Carter and Kohn}{1994}]{CarterKohn}
\begin{barticle}[author]
\bauthor{\bsnm{Carter},~\bfnm{C.~K.}\binits{C.~K.}} \AND
  \bauthor{\bsnm{Kohn},~\bfnm{R.}\binits{R.}}
(\byear{1994}).
\btitle{On Gibbs Sampling for State Space Models}.
\bjournal{Biometrika}
\bvolume{81}
\bpages{541-553}.
\end{barticle}
\endbibitem

\bibitem[\protect\citeauthoryear{El-Hay, Friedman and Kupferman}{2008}]{EFK}
\begin{binproceedings}[author]
\bauthor{\bsnm{El-Hay},~\bfnm{Tal}\binits{T.}},
  \bauthor{\bsnm{Friedman},~\bfnm{Nil}\binits{N.}} \AND
  \bauthor{\bsnm{Kupferman},~\bfnm{Raz}\binits{R.}}
(\byear{2008}).
\btitle{Gibbs Sampling in Factorized Continuous-Time {M}arkov Processes}.
In \bbooktitle{Proceedings of the Twenty-Fourth Conference Annual Conference on
  Uncertainty in Artificial Intelligence (UAI-08)}
\bpages{169--178}.
\bpublisher{AUAI Press}, \baddress{Corvallis, Oregon}.
\end{binproceedings}
\endbibitem

\bibitem[\protect\citeauthoryear{Fan and Shelton}{2008}]{FaSh}
\begin{binproceedings}[author]
\bauthor{\bsnm{Fan},~\bfnm{Yu}\binits{Y.}} \AND
  \bauthor{\bsnm{Shelton},~\bfnm{Christian~R.}\binits{C.~R.}}
(\byear{2008}).
\btitle{Sampling for Approximate Inference in Continuous Time {B}ayesian
  Networks}.
In \bbooktitle{Tenth International Symposium on Artificial Intelligence and
  Mathematics}.
\end{binproceedings}
\endbibitem

\bibitem[\protect\citeauthoryear{Fan, Xu and Shelton}{2010}]{FaXuSh}
\begin{barticle}[author]
\bauthor{\bsnm{Fan},~\bfnm{Yu}\binits{Y.}},
  \bauthor{\bsnm{Xu},~\bfnm{Jing}\binits{J.}} \AND
  \bauthor{\bsnm{Shelton},~\bfnm{Christian~R.}\binits{C.~R.}}
(\byear{2010}).
\btitle{Importance Sampling for Continuous Time {B}ayesian Networks}.
\bjournal{Journal of Machine Learning Research}
\bvolume{11}
\bpages{2115--2140}.
\end{barticle}
\endbibitem

\bibitem[\protect\citeauthoryear{Fr{\"u}hwirth-Schnatter}{1994}]{Fru-Sch}
\begin{barticle}[author]
\bauthor{\bsnm{Fr{\"u}hwirth-Schnatter},~\bfnm{Sylvia}\binits{S.}}
(\byear{1994}).
\btitle{Data augmentation and dynamic linear models}.
\bjournal{Journal of Time Series Analysis}
\bvolume{15}
\bpages{183--202}.
\bdoi{10.1111/j.1467-9892.1994.tb00184.x}
\end{barticle}
\endbibitem

\bibitem[\protect\citeauthoryear{Golightly, Henderson and
  Sherlock}{2015}]{GolHendSher2015}
\begin{barticle}[author]
\bauthor{\bsnm{Golightly},~\bfnm{Andrew}\binits{A.}},
  \bauthor{\bsnm{Henderson},~\bfnm{DanielA.}\binits{D.}} \AND
  \bauthor{\bsnm{Sherlock},~\bfnm{Chris}\binits{C.}}
(\byear{2015}).
\btitle{Delayed acceptance particle {MCMC} for exact inference in stochastic
  kinetic models}.
\bjournal{Statistics and Computing}
\bvolume{25}
\bpages{1039-1055}.
\bdoi{10.1007/s11222-014-9469-x}
\end{barticle}
\endbibitem

\bibitem[\protect\citeauthoryear{Golightly and Wilkinson}{2011}]{GolWilk2011}
\begin{barticle}[author]
\bauthor{\bsnm{Golightly},~\bfnm{Andrew}\binits{A.}} \AND
  \bauthor{\bsnm{Wilkinson},~\bfnm{Darren~J.}\binits{D.~J.}}
(\byear{2011}).
\btitle{Bayesian parameter inference for stochastic biochemical network models
  using particle {M}arkov chain {M}onte {C}arlo}.
\bjournal{Interface Focus}.
\bdoi{10.1098/rsfs.2011.0047}
\end{barticle}
\endbibitem

\bibitem[\protect\citeauthoryear{{Golightly} and
  {Wilkinson}}{2014}]{GolWilk2014}
\begin{barticle}[author]
\bauthor{\bsnm{{Golightly}},~\bfnm{A.}\binits{A.}} \AND
  \bauthor{\bsnm{{Wilkinson}},~\bfnm{D.~J.}\binits{D.~J.}}
(\byear{2014}).
\btitle{Bayesian inference for {M}arkov jump processes with informative
  observations}.
\bjournal{ArXiv e-prints}.
\end{barticle}
\endbibitem

\bibitem[\protect\citeauthoryear{Lauritzen}{2001}]{Lauritzen01causalinference}
\begin{barticle}[author]
\bauthor{\bsnm{Lauritzen},~\bfnm{Steffen~L}\binits{S.~L.}}
(\byear{2001}).
\btitle{Causal inference from graphical models}.
\bjournal{Complex stochastic systems}
\bpages{63--107}.
\end{barticle}
\endbibitem

\bibitem[\protect\citeauthoryear{Miasojedow and Niemiro}{2016}]{homo}
\begin{barticle}[author]
\bauthor{\bsnm{Miasojedow},~\bfnm{Blazej}\binits{B.}} \AND
  \bauthor{\bsnm{Niemiro},~\bfnm{Wojciech}\binits{W.}}
(\byear{2016}).
\btitle{Geometric ergodicity of Rao and Teh’s algorithm for homogeneous
  Markov jump processes}.
\bjournal{Statistics \& Probability Letters}
\bvolume{113}
\bpages{1 - 6}.
\bdoi{http://dx.doi.org/10.1016/j.spl.2016.02.002}
\end{barticle}
\endbibitem

\bibitem[\protect\citeauthoryear{Miasojedow et~al.}{2014}]{CTBNMet2014}
\begin{barticle}[author]
\bauthor{\bsnm{Miasojedow},~\bfnm{Blazej}\binits{B.}},
  \bauthor{\bsnm{Niemiro},~\bfnm{Wojciech}\binits{W.}},
  \bauthor{\bsnm{Noble},~\bfnm{John}\binits{J.}} \AND
  \bauthor{\bsnm{Opalski},~\bfnm{Krzysztof}\binits{K.}}
(\byear{2014}).
\btitle{Metropolis-type algorithms for Continuous Time Bayesian Networks}.
\bjournal{arXiv preprint arXiv:1403.4035}.
\end{barticle}
\endbibitem

\bibitem[\protect\citeauthoryear{Nodelman, Shelton and Koller}{2002a}]{Nod1}
\begin{binproceedings}[author]
\bauthor{\bsnm{Nodelman},~\bfnm{Uri}\binits{U.}},
  \bauthor{\bsnm{Shelton},~\bfnm{Christian~R}\binits{C.~R.}} \AND
  \bauthor{\bsnm{Koller},~\bfnm{Daphne}\binits{D.}}
(\byear{2002}a).
\btitle{Continuous time Bayesian networks}.
In \bbooktitle{Proceedings of the Eighteenth conference on Uncertainty in
  artificial intelligence}
\bpages{378--387}.
\end{binproceedings}
\endbibitem

\bibitem[\protect\citeauthoryear{Nodelman, Shelton and Koller}{2002b}]{Nod2}
\begin{binproceedings}[author]
\bauthor{\bsnm{Nodelman},~\bfnm{Uri}\binits{U.}},
  \bauthor{\bsnm{Shelton},~\bfnm{Christian~R}\binits{C.~R.}} \AND
  \bauthor{\bsnm{Koller},~\bfnm{Daphne}\binits{D.}}
(\byear{2002}b).
\btitle{Learning continuous time {B}ayesian networks}.
In \bbooktitle{Proceedings of the Nineteenth conference on Uncertainty in
  Artificial Intelligence}
\bpages{451--458}.
\bpublisher{Morgan Kaufmann Publishers Inc.}
\end{binproceedings}
\endbibitem

\bibitem[\protect\citeauthoryear{Rao and Teh}{2012}]{rao2012mcmc}
\begin{binproceedings}[author]
\bauthor{\bsnm{Rao},~\bfnm{Vinayak}\binits{V.}} \AND
  \bauthor{\bsnm{Teh},~\bfnm{Yee~W}\binits{Y.~W.}}
(\byear{2012}).
\btitle{{MCMC} for continuous-time discrete-state systems}.
In \bbooktitle{Advances in Neural Information Processing Systems}
\bpages{701--709}.
\end{binproceedings}
\endbibitem

\bibitem[\protect\citeauthoryear{Rao and Teh}{2013}]{RaoTeh2013a}
\begin{barticle}[author]
\bauthor{\bsnm{Rao},~\bfnm{Vinayak}\binits{V.}} \AND
  \bauthor{\bsnm{Teh},~\bfnm{Yee~W}\binits{Y.~W.}}
(\byear{2013}).
\btitle{Fast {MCMC} sampling for {M}arkov jump processes and extensions}.
\bjournal{Journal of Machine Learning Research}
\bvolume{14}
\bpages{3207--3232}.
\end{barticle}
\endbibitem

\bibitem[\protect\citeauthoryear{Roberts and
  Rosenthal}{2004}]{roberts2004general}
\begin{barticle}[author]
\bauthor{\bsnm{Roberts},~\bfnm{Gareth~O}\binits{G.~O.}} \AND
  \bauthor{\bsnm{Rosenthal},~\bfnm{Jeffrey~S}\binits{J.~S.}}
(\byear{2004}).
\btitle{General state space {M}arkov chains and {MCMC} algorithms}.
\bjournal{Probability Surveys}
\bvolume{1}
\bpages{20--71}.
\end{barticle}
\endbibitem

\bibitem[\protect\citeauthoryear{Schweder}{1970}]{Sch}
\begin{barticle}[author]
\bauthor{\bsnm{Schweder},~\bfnm{Tore}\binits{T.}}
(\byear{1970}).
\btitle{Composable markov processes}.
\bjournal{Journal of applied probability}
\bvolume{7}
\bpages{400--410}.
\end{barticle}
\endbibitem

\end{thebibliography}

\end{document}